%% file: main.tex
\documentclass[aps,pra,reprint,superscriptaddress,longbibliography,amsmath,amssymb]{revtex4-2}
\usepackage{iftex}
\ifPDFTeX\pdfoutput=1\fi

\usepackage[utf8]{inputenc}
\usepackage[english]{babel}
\usepackage[T1]{fontenc}
\usepackage{lmodern}
\usepackage{amsthm,mathtools}
\usepackage{microtype}
\usepackage[colorlinks=true,linkcolor=blue,citecolor=blue,urlcolor=blue]{hyperref}
\usepackage[dvipsnames]{xcolor}
\usepackage{tikz}
\usetikzlibrary{arrows.meta,positioning}
\usepackage{etoolbox}
\usepackage{needspace}
\AtBeginEnvironment{thebibliography}{\interlinepenalty=10000\relax}

\newtheorem{theorem}{Theorem}[section]
\newtheorem{corollary}[theorem]{Corollary}
\newtheorem{proposition}[theorem]{Proposition}
\theoremstyle{definition}
\newtheorem*{definition}{Definition}
\theoremstyle{plain}

\DeclareMathOperator{\Tr}{Tr}
\DeclareMathOperator{\CF}{CF}
\DeclareMathOperator{\sn}{sn}
\newcommand{\Z}{\mathbb Z}

\begin{document}

\title{Logical Bell Inequalities for Marginal Wigner Negativity}

\author{Sanchit Srivastava}
\email[]{sanchit.srivastava@uwaterloo.ca}
\affiliation{Institute for Quantum Computing, University of Waterloo, Waterloo, Ontario, Canada N2L 3G1}
\affiliation{Department of Physics and Astronomy, University of Waterloo, Waterloo, Ontario, Canada N2L 3G1}
\author{Selman Ipek}
\affiliation{Institute of Theoretical Physics, Leibniz University Hannover, Germany}
\author{Cihan Okay}
\affiliation{Department of Mathematics, Bilkent University, Ankara, Turkey}
\author{Shohini Ghose}
\affiliation{Institute for Quantum Computing, University of Waterloo, Waterloo, Ontario, Canada N2L 3G1}
\affiliation{Department of Physics and Astronomy, University of Waterloo, Waterloo, Ontario, Canada N2L 3G1}
\affiliation{Department of Physics and Computer Science, Wilfrid Laurier University, Waterloo, Ontario, Canada N2L 3C5}
\hypersetup{
  pdftitle={Logical Bell Inequalities for Marginal Wigner Negativity},
  pdfauthor={Sanchit Srivastava, Selman Ipek, Cihan Okay, Shohini Ghose}
}

\begin{abstract}
We construct logical Bell inequalities whose violations are determined by
the Wigner function of a reduced quantum state.  For odd-prime qudits,
each phase-space point of any nonempty proper subsystem defines an
inequality, violated precisely when the corresponding Gross Wigner value
is negative.  This correspondence allows us to quantify the contextuality
that a reduced state guarantees for the composite system.  The normalized
violation gives a lower bound on the contextual fraction in the complete
Weyl measurement scenario, valid for every joint state with that marginal,
irrespective of its correlations.  We identify a family of product states
that saturates this bound.  For these states, one marginal Wigner value
fixes the contextual fraction of the entire system.
\end{abstract}

\maketitle

\section{Introduction}
\label{sec:introduction}

Clifford stabilizer circuits can be simulated efficiently on a classical
computer \cite{AaronsonGottesman2004,HostensEtAl2005}.  Injecting suitable
nonstabilizer states, known as magic states, into these circuits enables the
implementation of the non-Clifford gates needed for universal quantum computation
\cite{BravyiKitaev2005}.  For odd-prime qudits, the discrete Wigner
representation in Ref.~\cite{Gross2006} assigns nonnegative phase-space
functions to stabilizer preparations and measurements.  Clifford transformations permute the
phase-space points and therefore preserve this nonnegativity
\cite{Gross2006}.  When the input qudits are
prepared independently and have nonnegative Wigner functions, a classical
computer can efficiently reproduce the circuit's output statistics by sampling
phase-space points.  This simulation also applies to nonstabilizer input
states with nonnegative Wigner functions.
To achieve a quantum computational speedup in this setting, the injected states
must have Wigner negativity \cite{VeitchEtAl2012}.  Wigner
sum negativity quantifies this resource by summing the magnitudes of the
negative Wigner values.  It is a magic monotone: it cannot increase on average
under stabilizer operations \cite{VeitchEtAl2014}.

A measurement context consists of observables that can be measured jointly.
A Kochen--Specker noncontextual model reproduces the statistics of all
contexts using a probability distribution over consistent outcome assignments
\cite{KochenSpecker1967,AbramskyBrandenburger2011}.  Each assignment gives
an observable the same outcome in every context and respects the relations
among jointly measured observables.  Statistics that admit no such model
are contextual.  Which statistics the model must reproduce depends on the
observables and contexts included in the measurement scenario.  We consider
the complete Weyl scenario, which contains every Weyl observable and every
maximal commuting context.  For systems of at least two odd-prime qudits,
the authors of Ref.~\cite{DelfosseEtAl2017} proved that a state is contextual
in this scenario exactly when its Gross-Wigner function takes a negative
value.  For two or more qubits, the complete Pauli scenario is
state-independently contextual, so contextuality in that scenario does not
distinguish magic states from stabilizer states
\cite{Peres1990,Mermin1990,OkayEtAl2017}.

We ask how much contextuality of a composite system can be certified from
one of its reduced states.  The reduced state's Wigner function is obtained
by summing the composite Wigner function over the phase-space coordinates
of the other subsystem \cite{Gross2006}.  A negative marginal value
therefore certifies negativity of the joint state.  In
Ref.~\cite{HowardEtAl2014}, the authors established, for odd-prime qudits,
an equivalence between single-qudit Wigner negativity and violation of
their two-qudit stabilizer noncontextuality inequalities.  We extend this
marginal test to arbitrary nonempty proper subsystems and obtain a
quantitative lower bound on the contextuality of the composite state.

We formulate this extension using logical Bell inequalities
\cite{AbramskyHardy2012}, which translate logical constraints on measurement
outcomes into probability bounds for noncontextual models.  Each phase-space
point of the chosen subsystem labels one inequality.  For a subsystem of
$k$ qudits with dimension $D_A=d^k$ and reduced
state $\sigma$, we write $W_\sigma(u)$ for its Gross Wigner function at the
phase-space point $u$.  The selected-event probabilities have a
noncontextual lower bound of one and a quantum sum $1+D_AW_\sigma(u)$
(Theorem~\ref{thm:marginal-lbi}).  Thus the inequality is violated exactly
when the chosen marginal Wigner value is negative.

The size of this violation also tells us how much contextuality the
marginal certifies.  For a joint state $\rho$ with reduced state $\sigma$,
let $e_\rho$ denote the full collection of Weyl measurement probabilities.
Its contextual fraction $\CF(e_\rho)$ is the smallest contextual weight in
a convex decomposition into a noncontextual model and a compatible
remainder \cite{AbramskyBarbosaMansfield2017}.  We show that
\begin{equation}
 \CF(e_\rho)\geq[-D_AW_\sigma(u)]_+,
 \label{eq:intro-cf-bound}
\end{equation}
where $[t]_+=\max\{t,0\}$ (Corollary~\ref{cor:marginal-cf-bound}).  This
bound holds for every joint state with the specified marginal, including
correlated states.  For a fixed subsystem, the smallest Wigner value gives
the strongest bound in this family.  If several subsystem marginals are
available, we take the largest of their bounds.
Evaluating one such probability sum uses $D_A+1$ Weyl measurement settings
on the chosen subsystem, each paired with the same fixed setting on the
other subsystem.  This number of settings is independent of the number
of remaining qudits.

The noncontextual bound requires each global assignment to give a valid
joint outcome in every context of the complete Weyl scenario.
The selected contexts by themselves admit a noncontextual model for every
compatible probability table.  A violation therefore rules out extending
those selected-context statistics to a noncontextual model on all Weyl
contexts.

We identify a family of product states that attains the marginal lower
bound and therefore minimizes the contextual fraction among all joint
states with the same reduced states.  The chosen marginal has one
nonpositive Wigner value and equal positive values elsewhere, while the
other subsystem can be in any
Wigner-nonnegative state.  For this family, a single marginal Wigner value determines the
contextual fraction of the composite state exactly.  Correlations can
increase the contextual fraction while leaving both reduced states
unchanged, as shown by the separable preparations in
Appendix~\ref{app:fixed-marginals}.

Section~\ref{sec:framework} introduces the
Weyl measurement setting and noncontextual models.  Section~\ref{sec:lbi}
derives the marginal logical Bell inequality and contextual-fraction bound,
examines the choice of subsystem, and explains why the noncontextual bound
requires the complete Weyl scenario.
Section~\ref{sec:exact-cf} constructs the product family for which marginal
Wigner negativity determines the contextual fraction exactly.
Section~\ref{sec:discussion} discusses the interpretation and scope of the
results.  Appendix~\ref{app:fixed-marginals} gives a special construction of
composite states with the same reduced states and different contextual
fractions.

\section{Weyl measurements and noncontextual models}
\label{sec:framework}

\subsection{Weyl outcomes and Wigner functions}

Discrete Wigner representations recover measurement probabilities by
summing over phase-space lines \cite{Wootters1987,GibbonsEtAl2004}.
For the Weyl measurements considered here, the summation regions are the
subsets of phase space that label complete joint outcomes.

Let $d$ be an odd prime, let $\omega=e^{2\pi i/d}$, and write $\Z_d$ for the
ring of integers modulo $d$.  Since $d$ is prime, this ring is a field.
On one qudit, the shift and phase operators act on the computational basis as
\begin{equation}
 X\lvert j\rangle=\lvert j+1\rangle,
 \qquad
 Z\lvert j\rangle=\omega^j\lvert j\rangle,
\end{equation}
where $j\in\Z_d$ and addition is modulo $d$.  For $r$ qudits, the
Hilbert-space dimension is $D_r=d^r$.  We label the Weyl operators by vectors
in the discrete symplectic phase space $V_r=\Z_d^{2r}$
\cite{HostensEtAl2005,Gross2006}.  If $v=(p,q)$ and $w=(p',q')$, with
$p,q,p',q'\in\Z_d^r$, their symplectic pairing is
\begin{equation}
 [v,w]=p\mathbin{\cdot}q'-p'\mathbin{\cdot}q.
 \label{eq:symplectic-form}
\end{equation}
Each label specifies a tensor product of shifts and phases.  With $1/2$
denoting the inverse of $2$ in $\Z_d$, we fix the Weyl convention as
\begin{equation}
 T_r(v)=\bigotimes_{i=1}^r
  \omega^{p_i q_i/2} X_i^{p_i} Z_i^{q_i},
 \label{eq:weyl-definition}
\end{equation}
so that
\begin{equation}
 T_r(v) T_r(w)=\omega^{-[v,w]/2} T_r(v+w).
 \label{eq:weyl-product}
\end{equation}
The symplectic pairing therefore determines which operators can be measured
jointly: $T_r(v)$ and $T_r(w)$ commute exactly when $[v,w]=0$.  In that case
their product is $T_r(v+w)$ \cite{HostensEtAl2005,Gross2006}.

Measuring a Weyl operator means measuring its spectral projectors and using
its eigenvalues as outcome labels.  Maximal joint Weyl measurements are
described by Lagrangian subspaces.

\begin{definition}[Lagrangian subspace]
An $r$-dimensional subspace $L\subset V_r$ is Lagrangian if
$[\ell,\ell']=0$ for every $\ell,\ell'\in L$.  The corresponding operators
$\{T_r(\ell):\ell\in L\}$ form a maximal commuting Weyl family
\cite{Gross2006}.
\end{definition}

Choose a basis $\ell_1,\ldots,\ell_r$ of $L$.  A joint outcome
assigns eigenvalues $\omega^{s_j}$ to the commuting operators
$T_r(\ell_j)$, with $s_j\in\Z_d$; the Weyl product rule then fixes the
eigenvalue of every $T_r(\ell)$ with $\ell\in L$.

We can represent these joint eigenvalues by a point in symplectic phase
space.  Choose dual vectors $h_1,\ldots,h_r$ with $[h_i,\ell_j]=\delta_{ij}$ and set
$a=\sum_i s_i h_i$.  Then $[a,\ell_j]=s_j$, and the eigenvalue of every
$T_r(\ell)$ in the commuting family is $\omega^{[a,\ell]}$.
Nondegeneracy guarantees that such dual vectors exist.  Changing the dual
vectors while retaining these pairings changes $a$ only by a vector in $L$:
the subspace pairing to zero with all of $L$ is $L$ itself.
Consequently, all points in
\begin{equation*}
 C=a+L=\{a+\ell:\ell\in L\}
\end{equation*}
label the same complete joint outcome, and points in different cosets label
different outcomes.  We therefore use cosets to label joint outcomes: each
coset specifies the eigenvalues of the entire commuting family without
choosing a basis or a particular phase-space representative.

There are $D_r$ distinct cosets of $L$, corresponding to the $D_r$ rank-one
outcomes of this maximal joint measurement.  The projector for outcome
$C=a+L$ is
\begin{equation}
 \Pi_C=\frac{1}{D_r}\sum_{\ell\in L}
   \omega^{-[a,\ell]}T_r(\ell).
 \label{eq:stabilizer-projector}
\end{equation}
Indeed, for every $\ell\in L$,
$T_r(\ell)\Pi_C=\omega^{[a,\ell]}\Pi_C$, which recovers the specified
joint eigenvalues \cite{HostensEtAl2005,Gross2006}.

The Gross Wigner representation assigns a quasiprobability to each point of
$V_r$.  For a state $\tau$ on the $r$ qudits and a point $u\in V_r$, define
the phase-point operator $A_u$ and the Wigner value $W_\tau(u)$ by
\begin{equation}
 \begin{aligned}
  A_u&=\frac{1}{D_r}\sum_{v\in V_r}\omega^{-[u,v]}T_r(v),\\
  W_\tau(u)&=\frac{1}{D_r}\Tr(A_u\tau).
 \end{aligned}
 \label{eq:gross-definitions}
\end{equation}
The values $W_\tau(u)$ are real and sum to one, but individual values can be
negative \cite{Gross2006}.  The phase-point operators satisfy the algebraic
identities
\begin{equation}
 \begin{aligned}
  \Tr(A_u)&=1,
  &A_u^2&=I,\\
  \Tr(A_u A_x)&=D_r\,\delta_{u,x}.
 \end{aligned}
 \label{eq:phase-point-identities}
\end{equation}
These follow from the Weyl expansion in Eq.~\eqref{eq:gross-definitions}
\cite{Gross2006}.

Measurement probabilities are recovered by summing the Wigner values over
the cosets that label the outcomes:
\begin{subequations}
\label{eq:phase-space-identities}
\begin{equation}
 \Tr(\Pi_C\tau)=\sum_{x\in C}W_\tau(x).
 \label{eq:coset-probability}
\end{equation}
To see the normalization and phases explicitly, write $x=a+\ell$ in the
phase-point expansion.  The sum over the outcome coset is
\begin{equation*}
 \sum_{x\in a+L}A_x
 =\frac{1}{D_r}\sum_{v\in V_r}\omega^{-[a,v]}T_r(v)
       \sum_{\ell\in L}\omega^{-[\ell,v]}.
\end{equation*}
The inner character sum is $D_r$ when $v\in L$ and zero otherwise, since
$L$ is its own symplectic orthogonal complement.  Hence
\begin{equation*}
 \sum_{x\in a+L}A_x
 =\sum_{v\in L}\omega^{-[a,v]}T_r(v)=D_r\Pi_{a+L}.
\end{equation*}
Taking the expectation and dividing by $D_r$ gives
Eq.~\eqref{eq:coset-probability}
\cite{Gross2006}.

We can also sum outcome probabilities across different measurements by
counting how their cosets cover phase space.  We choose the measurements
so that the cosets through a fixed point $u$ cover the whole space and
overlap only at $u$; a symplectic spread provides this arrangement.

\begin{definition}[Symplectic spread]
A symplectic spread $\mathcal L$ of $V_r$ is a collection of $D_r+1$
Lagrangian subspaces whose nonzero vectors partition $V_r\setminus\{0\}$.
\end{definition}

Every nonidentity Weyl operator therefore belongs to exactly one of the
corresponding maximal commuting families.  Such spreads exist in $V_r$ by
the standard finite-field construction \cite{GibbonsEtAl2004}.

Fix a point $u\in V_r$ and select the outcome $u+L$ from each measurement
in the spread.  When their projectors in Eq.~\eqref{eq:stabilizer-projector}
are summed, every nonzero Weyl label occurs once, while the identity term
occurs $D_r+1$ times.  Comparing with the Fourier expression for $A_u$ gives
\begin{equation}
 \sum_{L\in\mathcal L}\Pi_{u+L}=I+A_u.
 \label{eq:spread-projector-sum}
\end{equation}
\end{subequations}
For a state $\tau$, taking the expectation value gives a sum of outcome
probabilities equal to $1+D_r W_\tau(u)$.

For one qutrit, $V_1=\Z_3^2$ consists of nine points.  The four lines
through the origin form a spread:
\begin{equation}
 \mathcal L_3=\bigl\{\langle(1,0)\rangle,\langle(0,1)\rangle,
                 \langle(1,1)\rangle,\langle(1,2)\rangle\bigr\},
 \label{eq:qutrit-spread}
\end{equation}
where $\langle v\rangle=\{0,v,2v\}$.  Each line and its two parallel
cosets describe the three outcomes of one stabilizer measurement; the four
measurements have mutually unbiased eigenbases.  Selecting the outcome
through $u=0$ in each measurement selects all four lines in the spread.
They share the origin and partition the other eight points, so the sum of
their probabilities is $1+3W_\tau(0)$.

For a composite system, the phase-space representation respects the
tensor-product structure of the Hilbert space.  Consider $n=k+m$ qudits
split into nonempty subsystems $A$ and $B$, with $k,m\geq1$.  Write
$V_A\cong V_k$ and $V_B\cong V_m$, so that $V_n=V_A\oplus V_B$.  We use $T_A$ and
$T_B$ for the Weyl operators on the corresponding subsystems.  For
$v_A,u\in V_A$ and $v_B,v\in V_B$, the symbols $A_u$ and $A_v$ likewise denote
their respective phase-point operators.  If $\rho_A$ and $\rho_B$ are states
on the two subsystems, then
\begin{equation}
 \begin{aligned}
 T_n(v_A\oplus v_B)&=T_A(v_A)\otimes T_B(v_B),\\
 A_{u\oplus v}&=A_u\otimes A_v,\\
 W_{\rho_A\otimes\rho_B}(u,v)
   &=W_{\rho_A}(u)W_{\rho_B}(v).
 \end{aligned}
 \label{eq:tensor-factorization}
\end{equation}
Here $(u,v)$ is shorthand for the composite phase-space point $u\oplus v$.
The first identity follows directly from Eq.~\eqref{eq:weyl-definition}.
Factoring the Fourier sum in Eq.~\eqref{eq:gross-definitions} then gives
the phase-point identity, whose expectation in a product state gives the
Wigner factorization \cite{Gross2006}.

For a general composite state $\rho$, set $\sigma=\Tr_B\rho$.
Summing the phase-point operators over $B$ removes all nonidentity Weyl
terms in that factor.  The resulting identity is
\begin{equation}
 W_\sigma(u)=\sum_{v\in V_B}W_\rho(u,v).
 \label{eq:wigner-marginalization}
\end{equation}
This holds for correlated as well as product states \cite{Gross2006}.
Ignoring the measurement outcomes on $B$ therefore corresponds to summing
over its phase-space coordinates in the Wigner representation.

\subsection{Logical Bell inequalities}
\label{sec:logical-framework}

The logical Bell construction in Ref.~\cite{AbramskyHardy2012} bounds the
sum of event probabilities in a noncontextual model when no
context-independent outcome assignment can make all the events occur.  Each
event is testable within one measurement context, although the different
contexts need not be jointly measurable.

For a measurement context $\mathsf C$, a joint outcome $s$ specifies a result
$s(a)$ for every measurement $a\in\mathsf C$.  The atomic proposition $[a=s(a)]$
states that measurement $a$ gives the result $s(a)$.  A complete joint outcome
requires all these statements to hold, so it is represented by the
conjunction
\begin{equation}
 \phi_{\mathsf C,s}:=\bigwedge_{a\in\mathsf C}[a=s(a)],
 \label{eq:joint-outcome-proposition}
\end{equation}
where $\wedge$ denotes logical AND.  For a context consisting of measurements
$a$ and $b$, for example, the joint outcome states that $a$ gives $s(a)$ and
$b$ gives $s(b)$.

An event $E$ is a set of complete joint outcomes within one context.  It occurs when
the observed outcome is any one of those in $E$, so its proposition is the
logical OR of the corresponding joint-outcome statements.  If
$p_{\mathsf C}(s)$ is the probability of observing $s$ when context
$\mathsf C$ is measured, then
\begin{equation}
 \phi_{\mathsf C,E}:=\bigvee_{s\in E}\phi_{\mathsf C,s},
 \qquad
 p(\phi_{\mathsf C,E})=\sum_{s\in E}p_{\mathsf C}(s),
 \label{eq:event-proposition}
\end{equation}
where $\vee$ denotes logical OR.  The probabilities add because distinct
joint outcomes of one context are mutually exclusive.
An event that specifies only some measured values is represented by all
complete outcomes consistent with those values.  For example, fixing the
outcome on $A$ and leaving $B$ unrestricted means summing over every
compatible $B$ outcome in that same context.

A deterministic global assignment chooses one outcome for every measurement
in the scenario, with the same choice in every context containing that
measurement.  Consider $N$ event propositions $\phi_i$ that no global
assignment can satisfy simultaneously.  Every global assignment then makes
at least one proposition false and hence satisfies at most $N-1$.

A noncontextual model reproduces the outcome distributions of all contexts
using a single probability distribution over these global assignments
\cite{AbramskyBrandenburger2011}.  The probability of $\phi_i$ is the total
weight of assignments that make it true.  Summing the event probabilities
therefore gives the average number of satisfied propositions: an assignment
satisfying several propositions contributes once for each of them.  The
probability of their logical OR would instead count that assignment only
once.  The counting bound, rather than a union probability, gives the
logical Bell inequality.  Since every
assignment satisfies at most $N-1$, this average obeys
\begin{equation}
 \sum_{i=1}^{N}p(\phi_i)\leq N-1.
 \label{eq:general-logical-bell}
\end{equation}
More generally, the tight bound for this sum is the largest number of the
chosen propositions that one global assignment can satisfy.

Violation of this bound rules out a noncontextual model for the measured
probability table.  The construction applies both to spatial Bell
experiments and to general contextuality scenarios, where the contexts are
sets of compatible measurements \cite{AbramskyHardy2012}.

\subsection{Complete Weyl statistics and contextual fraction}
\label{sec:complete-model}

We consider $n\geq2$ qudits of odd-prime local dimension $d$.

\begin{definition}[Complete Weyl scenario]
The complete $n$-qudit Weyl measurement scenario contains every Weyl
observable $T_n(v)$, with $v\in V_n$, and every maximal commuting context.
A context is labelled by a Lagrangian subspace $K\subset V_n$ and consists
of the observables $\{T_n(v):v\in K\}$.  Its joint outcomes are the cosets
$C=a+K$, with projectors $\Pi_C$.  An observable appearing in several
contexts is the same measurement in each
\cite{Gross2006,DelfosseEtAl2017}.
\end{definition}

The state $\rho$ assigns to each full-system context the distribution
\begin{equation}
 p_\rho^K(C)=\Tr(\Pi_C\rho),\qquad C\in V_n/K.
 \label{eq:born-context-table}
\end{equation}
Its empirical model is the collection
$e_\rho=(p_\rho^K)_K$ over all maximal contexts.  An event $E$ in $K$ has
probability $p_\rho(E)=\sum_{C\in E}p_\rho^K(C)$.  The proposition
$\phi_{K,E}$ is true exactly when the event $E$ occurs, so its probability
is $p_\rho(\phi_{K,E})=p_\rho(E)$.  Thus $p_\rho$ refers to entries or sums of
entries of $e_\rho$, rather than to a separate distribution over contexts.
When the distribution is viewed on all tuples of Weyl eigenvalues in $K$,
tuples that do not form a valid coset outcome have probability zero.

A deterministic global assignment uses a single function
$\nu:V_n\to\Z_d$ to assign the eigenvalue $\omega^{\nu(v)}$ to each
observable $T_n(v)$ \cite{AbramskyBrandenburger2011}.
To give a valid joint outcome in every context $K$, its assigned values
must agree with one of that context's coset outcomes.

For a coset $C=a+K$, let $\phi_{K,C}$ state that all the assigned Weyl
eigenvalues match those of $C$.  The proposition $\phi_K$
includes all valid joint outcomes of $K$, so it states that one of them
occurs:
\begin{equation}
 \begin{aligned}
  \phi_{K,C}&:=\bigwedge_{v\in K}
    [T_n(v)=\omega^{[a,v]}],\\
  \phi_K&:=\bigvee_{C\in V_n/K}\phi_{K,C}.
 \end{aligned}
 \label{eq:probability-one-proposition}
\end{equation}
Each bracketed Weyl-value proposition is shared by every context containing
that observable.  The conditions $\phi_K$ therefore constrain the same
assignment across overlapping contexts.

A joint measurement of $K$ always returns one of its coset outcomes, so
$p(\phi_K)=1$.  In a noncontextual model, this probability is the total weight
of deterministic assignments satisfying $\phi_K$.  Since that weight is one,
no assignment violating $\phi_K$ can have positive weight.  Applying this
argument to every context shows that each assignment of nonzero weight
satisfies all the propositions $\phi_K$.

For each context $K$, truth of $\phi_K$ gives a coset $C=a+K$ with
$\nu(v)=[a,v]$ for every $v\in K$; the representative $a$ can depend on
$K$.  Linearity of the symplectic pairing in its second argument makes
$\nu$ additive within that context.  Every commuting pair lies in a maximal
context, so simultaneous truth of all $\phi_K$ implies
\begin{equation}
 \nu(v+w)=\nu(v)+\nu(w)
 \qquad\text{whenever }[v,w]=0.
 \label{eq:commuting-additivity}
\end{equation}

For commuting $v,w$, Eq.~\eqref{eq:weyl-product} identifies $T_n(v+w)$
with the product $T_n(v)T_n(w)$.  Its eigenvalue can therefore be obtained
either by measuring $T_n(v+w)$ directly or by jointly measuring the two
factors and multiplying their outcomes.  These procedures have the same
outcome projectors.  Implementation noncontextuality requires their
conditional outcome probabilities to agree for each hidden-variable value.
In a deterministic model, this requires equality of the direct and product
outcomes.  In Ref.~\cite[Proposition~3]{DelfosseEtAl2017}, the authors derive
the commuting-pair additivity condition from this operational requirement.

For odd $d$ and $n\geq2$, it is shown in Ref.~\cite[Lemma~1]{DelfosseEtAl2017}
that every such assignment has the form
\begin{equation}
 \nu_x(v)=[x,v]
 \qquad (x,v\in V_n).
 \label{eq:deterministic-values}
\end{equation}
The same phase-space point $x$ now determines the joint outcome $x+K$ in
every context $K$.  Noncontextual models of the complete scenario are therefore
probability mixtures of these assignments.

A state $\rho$ with a nonnegative Gross Wigner function supplies such a
probability distribution explicitly.  Sample $x$ with probability
$W_\rho(x)$ and return the coset $x+K$ in context $K$.
Equation~\eqref{eq:coset-probability} gives the probability of any outcome
$C$ as
$p_\rho(C)=\sum_{x\in C}W_\rho(x)=\Tr(\Pi_C\rho)$.
This reproduces every maximal context in the complete Weyl scenario, and
nonmaximal commuting measurements follow by coarse-graining
\cite{Gross2006,DelfosseEtAl2017}.

We can interpret a noncontextual model in terms of repeated experimental
runs.  On each run, we sample a global assignment and report the outcomes
it specifies for the measured context.  The assignment may vary from run
to run, but its sampling distribution is independent of which context is
measured.  For a contextual probability table, we ask for the largest
fraction of runs that can be generated in this way while reproducing the
outcome probabilities in every context.

We separate the statistics into a noncontextual contribution and a
remainder on the same measurement scenario:
\begin{equation}
 e=(1-\lambda)e_{\mathrm{NC}}+\lambda e',
 \qquad 0\leq\lambda\leq1,
 \label{eq:cf-decomposition}
\end{equation}
Here $e_{\mathrm{NC}}$ is noncontextual, and $e'$ is a compatible probability
model: it assigns a normalized outcome distribution to each context, with
agreeing marginals wherever contexts overlap.  In a statistical simulation,
a fraction $1-\lambda$ of runs uses the sampled global assignments, and the
remaining fraction $\lambda$ uses $e'$.  The same mixing weight applies in
every context, and the combined statistics reproduce every entry of $e$.

The contextual fraction $\CF(e)$ is the minimum possible value of
$\lambda$ \cite{AbramskyBarbosaMansfield2017}.  It is therefore the smallest
fraction of runs that must be supplied by the remainder after accounting
for as much of the statistics as possible with a mixture of global
assignments.  The optimization ranges over all noncontextual probability
models and compatible remainders on the same measurement scenario, without
requiring quantum-state realizations of the two components.  For the
complete Weyl statistics $e_\rho$, a nonnegative Wigner function supplies a
fully noncontextual model, and the converse follows from the phase-space
characterization above.  Hence $\CF(e_\rho)=0$ exactly when $W_\rho$ is
everywhere nonnegative \cite{DelfosseEtAl2017}.

\section{Marginal logical Bell inequalities}
\label{sec:lbi}

\subsection{Logical Bell violations and marginal Wigner negativity}

We test a Wigner value of a proper subsystem using probabilities of joint
measurements on the composite system.  Split $n=k+m$ qudits, with $k,m\geq1$, placing
$k$ qudits in the chosen subsystem $A$ and $m$ in its complement $B$.
Their phase spaces are
$V_A\cong\Z_d^{2k}$ and $V_B\cong\Z_d^{2m}$, with Hilbert-space dimensions
\begin{equation}
 D_A=d^k,
 \qquad
 D_B=d^m.
 \label{eq:block-dimensions}
\end{equation}
To test the marginal Wigner value at $u\in V_A$, choose a symplectic spread
$\mathcal L$ of $V_A$ and select the outcome $u+L$ from each context $L$.
These target cosets cover $V_A$, so no phase-space point can avoid all the
selected outcomes.  On $B$, fix a Lagrangian subspace $S\subset V_B$,
corresponding to one maximal commuting measurement.  Combining each $A$
measurement with this same $B$ measurement gives the joint maximal contexts
\begin{equation}
 K_L=L\oplus S.
\end{equation}
Let $E_L$ be the event that the $A$ outcome is $u+L$, irrespective of the
$B$ outcome.  Its probability is obtained by summing over all $B$ outcomes.
For an arbitrary composite state $\rho$ with
$\sigma=\Tr_B\rho$,
\begin{equation}
 \begin{aligned}
  p_\rho(E_L)
  &=\sum_{C_B\in V_B/S}
    \Tr\!\left[\rho(\Pi_{u+L}\otimes\Pi_{C_B})\right]\\
  &=\Tr(\Pi_{u+L}\sigma).
 \end{aligned}
 \label{eq:coarse-graining}
\end{equation}
The $B$-outcome projectors sum to $I_B$, leaving a probability fixed by the
reduced state $\sigma$.

In each context $K_L$, the target event $E_L$ and its complement
$E_L^{\mathrm c}$ partition the valid joint outcomes.  The complement
contains exactly those outcomes whose $A$ coset differs from $u+L$; the
$B$ outcome remains unrestricted.  A fine-grained outcome of $K_L$ is a pair of cosets
$(C_A,C_B)\in V_A/L\times V_B/S$.  Let
$\phi_{K_L,(C_A,C_B)}$ denote the conjunction in
Eq.~\eqref{eq:joint-outcome-proposition} that specifies its Weyl eigenvalues.
Let $\phi_L^{(u)}$ state that, in context $K_L$, the $A$ outcome is a coset
other than $u+L$, with any $B$ outcome allowed:
\begin{equation}
 \phi_L^{(u)}=
 \bigvee_{\substack{C_A\in V_A/L\\ C_A\neq u+L}}
 \;\bigvee_{C_B\in V_B/S}\phi_{K_L,(C_A,C_B)},
 \label{eq:spread-proposition}
\end{equation}
Thus $\phi_L^{(u)}$ expresses $E_L^{\mathrm c}$ and
$p(\phi_L^{(u)})=1-p(E_L)$.  In contrast, $\phi_K$ expresses that the eigenvalue
assignment is a valid outcome of $K$.

Let $K_1,\ldots,K_R$ be the remaining maximal contexts in the complete Weyl
scenario.  For each $K_j$, include the valid-outcome proposition
$\phi_{K_j}$ from Eq.~\eqref{eq:probability-one-proposition},
with $p(\phi_{K_j})=1$.
The full family consists of the $D_A+1$ propositions
$\phi_L^{(u)}$ and the $R$ propositions $\phi_{K_j}$.
These formulas use the same Weyl-value propositions on every overlap between
contexts.

For example, take $A$ to be one qutrit and choose the origin $u=(0,0)$.
Its phase space consists of nine points $(p,q)\in\Z_3^2$, with
coordinate arithmetic modulo three.  We choose the horizontal line
$L_1=\{(0,0),(1,0),(2,0)\}$.  This is a Lagrangian subspace: its vectors have
zero symplectic product with one another, and no larger subspace of
$V_A$ has this property.  It specifies a commuting Weyl measurement.
The three parallel lines $L_1$, $(0,1)+L_1$, and $(0,2)+L_1$ label its
three possible outcomes, partitioning the nine points as in
Fig.~\ref{fig:measurement-setup}(a).

We obtain an outcome probability by summing the marginal Wigner values
along its line, as in Eq.~\eqref{eq:coset-probability}.  This is the
discrete counterpart of integrating the continuous-variable Wigner
function along a line to obtain a quadrature probability density.
For the solid line through the origin in Fig.~\ref{fig:measurement-setup}(a),
\begin{equation*}
 p_\rho(E_{L_1})=W_\sigma(2,0)+W_\sigma(0,0)+W_\sigma(1,0).
\end{equation*}
Here $\sigma=\Tr_B\rho$: summing over the joint measurement's $B$
outcomes leaves precisely this line sum on $A$.  The proposition
$\phi_{L_1}^{(0)}$ accepts either of the two dashed-line outcomes and
excludes the solid-line outcome, so its probability is $1-p_\rho(E_{L_1})$.

The complete-scenario validity conditions let us label each deterministic
assignment by a phase-space point $x=(x_A,x_B)$, as in
Eq.~\eqref{eq:deterministic-values}.  Its outcome in the horizontal
measurement is the unique coset $x_A+L_1$ containing $x_A$.
For example, $x_A=(1,1)$ selects the upper dashed line in
Fig.~\ref{fig:measurement-setup}(a) and satisfies $\phi_{L_1}^{(0)}$.
The same point determines the outcome $x_A+L_i$ in every other
measurement.  To satisfy all four exclusion propositions, it must
therefore lie outside every target line $L_i$.  We leave $x_B$ free,
since the propositions accept every $B$ outcome.

We choose the vertical line $L_2$ and the two diagonal lines $L_3,L_4$
from Eq.~\eqref{eq:qutrit-spread} for the remaining measurements.
The four target lines cover the grid in
Fig.~\ref{fig:measurement-setup}(b), so no point $x_A$ avoids them all.
For instance, the point $(1,1)$ lies on $L_3$ and thus fails
$\phi_{L_3}^{(0)}$, although it satisfies the first exclusion.
Every noncentral point lies on exactly one target line and satisfies
the other three exclusion propositions.  Averaging over assignments
therefore bounds the sum of exclusion probabilities by three, or
equivalently gives a lower bound of one on the sum of the four
target-event probabilities.

The same lines also give the quantum probability sum.  Adding their
Wigner sums counts each noncentral point once and the origin four
times.  The sum over all nine points is one, leaving three additional
copies of $W_\sigma(0)$:
\begin{equation}
 \sum_{L\in\mathcal L_3}p_\rho(E_L)=1+3W_\sigma(0),
 \label{eq:qutrit-test}
\end{equation}
so a negative value at the origin lowers this sum below its
noncontextual bound.

\begin{figure}[t]
 \centering
 \input{figures/measurement-setup.tex}
 \caption{Qutrit measurement outcomes and target-line coverage at
 $u=(0,0)$, with coordinates modulo three.
 (a) The cosets of $L_1$ label one measurement's outcomes.  The solid
 target line is excluded by $\phi_{L_1}^{(0)}$; the two dashed lines
 are accepted.  The circled point $x_A=(1,1)$ selects the upper outcome.
 (b) The four target lines cover phase space, meeting only at $u$.
 The same circled point lies on $L_3$ and fails its exclusion.
 The multiplicities give the Wigner probability sum in
 Eq.~\eqref{eq:qutrit-test}.}
 \label{fig:measurement-setup}
\end{figure}
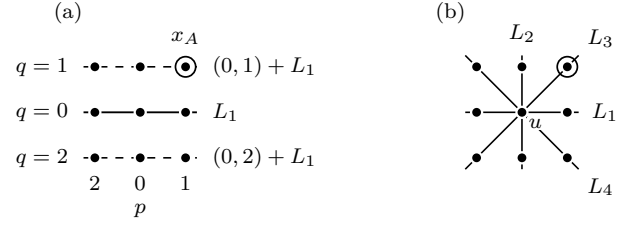

\begin{theorem}
\label{thm:marginal-lbi}
Let $d$ be an odd prime and $n=k+m$ with $k,m\geq1$.  For a symplectic spread
$\mathcal L$ of $V_A$, a point $u\in V_A$, and a Lagrangian subspace
$S\subset V_B$, define the contexts, events, and propositions as above.
The family of $D_A+1+R$ propositions $\phi_L^{(u)}$ and
$\phi_{K_j}$ is jointly unsatisfiable.  The associated logical
Bell inequality has a noncontextual
bound attained by a deterministic assignment of the complete $n$-qudit Weyl
scenario.  Using $p(\phi_{K_j})=1$, this inequality reduces to
\begin{equation}
 \sum_{L\in\mathcal L}p(E_L)\geq1.
 \label{eq:lbi}
\end{equation}
For every composite state $\rho$ with $A$-marginal
$\sigma=\Tr_B\rho$,
\begin{equation}
 \sum_{L\in\mathcal L}p_\rho(E_L)
 =1+D_A W_\sigma(u).
 \label{eq:probability-sum}
\end{equation}
The inequality is violated if and only if $W_\sigma(u)<0$.
\end{theorem}

\begin{proof}
Suppose a global truth assignment satisfied every $\phi_L^{(u)}$
and every $\phi_{K_j}$.  Each $\phi_L^{(u)}$ is a
disjunction of valid outcomes of $K_L$, so it also implies
$\phi_{K_L}$.  The
assignment would satisfy the probability-one proposition for every maximal
Weyl context and hence, by Eq.~\eqref{eq:deterministic-values}, would be
labelled by a point $x=(x_A,x_B)$.  Its $A$ outcome in $K_L$ would be
$x_A+L$.  Truth of $\phi_L^{(u)}$ would then require $x_A+L\neq u+L$, or
equivalently $x_A-u\notin L$, for every $L\in\mathcal L$.  This is
impossible because every vector of $V_A$ belongs to a member of the spread.
The family is therefore jointly unsatisfiable.

Applying Eq.~\eqref{eq:general-logical-bell} to these $D_A+1+R$ propositions
gives
\begin{equation}
 \sum_{L\in\mathcal L}p(\phi_L^{(u)})
 +\sum_{j=1}^{R}p(\phi_{K_j})\leq D_A+R.
 \label{eq:full-logical-bell}
\end{equation}
Since $p(\phi_{K_j})=1$, the probability-one terms cancel:
\begin{equation}
 \sum_{L\in\mathcal L}p(\phi_L^{(u)})\leq D_A.
 \label{eq:spread-logical-bell}
\end{equation}
Using $p(\phi_L^{(u)})=1-p(E_L)$ proves Eq.~\eqref{eq:lbi}.

To establish tightness, choose a complete-scenario assignment labelled by
$x=(x_A,x_B)$ with $x_A\neq u$.  The nonzero vector $x_A-u$ lies in exactly
one member of the spread, so exactly one target event $E_L$ occurs and
exactly $D_A$ of the propositions $\phi_L^{(u)}$ are true.  Every
$\phi_{K_j}$ is also true.  This assignment saturates both the
full logical Bell inequality
and the reduced bound.

For a quantum model, Eq.~\eqref{eq:coarse-graining} reduces the selected-event
sum to the expectation of the spread projectors on $A$.  Their sum is
$I+A_u$ by Eq.~\eqref{eq:phase-space-identities}, so
\begin{equation*}
 \begin{aligned}
  \sum_{L\in\mathcal L}p_\rho(E_L)
  &=\Tr\!\left[(I+A_u)\sigma\right]\\
  &=1+\Tr(A_u\sigma)=1+D_A W_\sigma(u).
 \end{aligned}
\end{equation*}
The second line uses normalization and Eq.~\eqref{eq:gross-definitions}.
Comparing this value with Eq.~\eqref{eq:lbi} proves the stated violation
criterion.
\end{proof}

The violation also bounds the contextual weight in the full probability
table $e_\rho$, as defined in Section~\ref{sec:complete-model}.

\begin{corollary}
\label{cor:marginal-cf-bound}
Under the hypotheses of Theorem~\ref{thm:marginal-lbi}, every composite state
$\rho$ with $A$-marginal $\sigma$ satisfies
\begin{equation}
 \CF(e_\rho)\geq[-\Tr(A_u\sigma)]_+
 =[-D_A W_\sigma(u)]_+,
 \label{eq:cf-bound}
\end{equation}
where $[t]_+=\max\{t,0\}$.
\end{corollary}

\begin{proof}
We use the convex-decomposition argument underlying the normalized-violation
bound in Ref.~\cite[Theorem~1(i)]{AbramskyBarbosaMansfield2017}.
Consider any decomposition
$e_\rho=(1-\lambda)e_{\mathrm{NC}}+\lambda e'$ as defined above.
For $e_{\mathrm{NC}}$, the full probability sum in
Eq.~\eqref{eq:full-logical-bell} is at most $D_A+R$.
For $e'$, it is at most $D_A+1+R$, since each of the $D_A+1+R$
probabilities is at most one.  Probabilities are linear under this
decomposition, so
\begin{equation*}
 \begin{aligned}
  &\sum_{L\in\mathcal L}p_\rho(\phi_L^{(u)})
   +\sum_{j=1}^{R}p_\rho(\phi_{K_j})\\
  &\quad\leq(1-\lambda)(D_A+R)+\lambda(D_A+1+R)\\
  &\quad=D_A+R+\lambda.
 \end{aligned}
\end{equation*}
For the quantum model, $p_\rho(\phi_{K_j})=1$ and
$p_\rho(\phi_L^{(u)})=1-p_\rho(E_L)$, so the same probability sum equals
\begin{equation*}
 D_A+1+R-\sum_{L\in\mathcal L}p_\rho(E_L).
\end{equation*}
Comparing these expressions and using Eq.~\eqref{eq:probability-sum} yields
\begin{equation*}
 \lambda\geq1-\sum_{L\in\mathcal L}p_\rho(E_L)
 =-D_A W_\sigma(u).
\end{equation*}
This bound holds for every admissible decomposition of $e_\rho$.
Minimizing $\lambda$ and using $\lambda\geq0$ proves Eq.~\eqref{eq:cf-bound}.
\end{proof}

For a qutrit marginal with $W_\sigma(0)=-1/3$, Eq.~\eqref{eq:qutrit-test}
gives a target-event probability sum of zero.  The target-event
probabilities are nonnegative, so each target event has
probability zero and each $\phi_L^{(u)}$ has probability
one.  The four complementary-event probabilities then sum to four,
exceeding their complete-scenario bound of three.  Corollary~\ref{cor:marginal-cf-bound}
gives $\CF(e_\rho)=1$ for any joint state with this marginal.

\Needspace{8\baselineskip}
\subsection{Choice of subsystem}
\label{sec:subsystem-choice}

Each nonempty proper subsystem gives a family of lower bounds.  For a
subsystem $A$, denote the strongest of its pointwise bounds by
\begin{equation}
 b_A(\rho)=\max_{u\in V_A}[-D_AW_{\rho_A}(u)]_+,
 \qquad \rho_A=\Tr_{\bar A}\rho.
 \label{eq:optimized-marginal-bound}
\end{equation}
Here $\bar A$ is the complementary subsystem.  If several reduced states
are available, they give the combined bound
\begin{equation}
 \CF(e_\rho)\geq\max_{A\in\mathcal A}b_A(\rho),
 \label{eq:multiple-marginal-bound}
\end{equation}
where $\mathcal A$ is the chosen collection of nonempty proper subsystems.

Keeping more qudits need not give a larger value of $b_A$.  For nested
subsystems $C\subset A=C\cup R$, with $C$ and $R$ disjoint, marginalization
and $D_A=D_CD_R$ give
\begin{equation}
 -D_CW_{\rho_C}(u)
 =\frac{1}{D_R}\sum_{v\in V_R}
        [-D_AW_{\rho_A}(u,v)].
 \label{eq:nested-marginal-bound}
\end{equation}
This sum has $D_R^2$ terms and can contain both signs.  For example,
if $\rho_A=\rho_C\otimes I_R/D_R$, every point of $V_R$ has Wigner value
$1/D_R^2$, and
\begin{equation}
 b_A(\rho)=\frac{b_C(\rho)}{D_R}.
 \label{eq:mixed-extension-bound}
\end{equation}
For a qutrit state $\sigma$ with $W_\sigma(0)=-1/3$, the three-qutrit
state $\sigma\otimes I_2/3\otimes I_3/3$, with $I_j$ the identity on
qudit $j$, therefore has strongest one-qutrit bound one and strongest
two-qutrit bound $1/3$.

The reverse behavior also occurs.  Take the two-qutrit state
$\rho_{\mathrm{corr}}(0)$ constructed in
Appendix~\ref{app:fixed-marginals} and append a maximally mixed third qutrit.
Each one-qutrit marginal has a nonnegative Wigner function, so all
one-qutrit bounds vanish.  The first two qutrits have a joint parity
expectation of $-1$, giving $b_{\{1,2\}}=1$.  Thus even after optimizing
over subsystems of a fixed size, the bound has no monotone ordering with
that size.  Choosing the best available marginal witness requires comparing
the dimension-scaled negative values, rather than always retaining the
largest proper subsystem.
A larger reduced density operator also determines all its smaller
marginals.  Access to it therefore allows their witnesses to be evaluated
as well; Eq.~\eqref{eq:multiple-marginal-bound} retains the strongest of
these available bounds.

\subsection{Selected contexts and complete-scenario extensions}
\label{sec:scenarios}

Every selected context combines a measurement on $A$ with the same
measurement $S$ on $B$.  Distinct spread members share only the zero Weyl
label, so their $A$ measurements share no nonidentity Weyl observable.
The selected contexts therefore overlap precisely on the fixed $B$
measurement:
\begin{equation}
 K_L\cap K_{L'}=S
 \qquad (L\neq L').
 \label{eq:context-intersection}
\end{equation}
A deterministic assignment to this family must choose one common $B$ outcome
$C_B\in V_B/S$, while its $A$ outcome $C_L\in V_A/L$ can be chosen freely
for each $L$.

\begin{proposition}
\label{thm:selected-complete}
Let $d$ be an odd prime and $n=k+m$ with $k,m\geq1$.  Fix a symplectic
spread $\mathcal L$ of $V_A$ and a Lagrangian subspace $S\subset V_B$, and
set $K_L=L\oplus S$.
\begin{enumerate}
 \item Every family of probability distributions on the selected contexts
 whose marginals on the common $B$ measurement agree admits a noncontextual
 model on those contexts.
 \item A deterministic assignment $\bigl(C_B,(C_L)_L\bigr)$ to the selected
 contexts extends to the complete Weyl scenario if and only if all its
 $A$-outcome cosets contain a common point $x_A\in V_A$, equivalently,
 \begin{equation}
  C_L=x_A+L
  \qquad\text{for every }L\in\mathcal L.
  \label{eq:common-point}
 \end{equation}
 The common $B$-coset $C_B$ is unrestricted.
\end{enumerate}
\end{proposition}

\begin{proof}
Let $p_L(C_L,C_B)$ be the supplied joint-outcome distribution in context
$K_L$.  The common $B$ marginal is
\begin{equation}
 q(C_B)=\sum_{C_L\in V_A/L}p_L(C_L,C_B),
 \label{eq:common-b-marginal}
\end{equation}
which is independent of $L$ by compatibility.  To construct a noncontextual
model, first sample $C_B$ with probability $q(C_B)$.  For an outcome with
$q(C_B)>0$, the data fix a conditional distribution
$p_L(C_L,C_B)/q(C_B)$ for each choice of $L$.  A hidden-variable assignment
needs an $A$ outcome for every alternative setting, even though only one
setting is measured in a run.  The supplied tables leave the joint
distribution of these alternative outcomes unspecified.  Choose their
product distribution conditional on $C_B$.  Together with $C_B$, the
sampled cosets specify one deterministic assignment, with probability
\begin{equation}
 \mu\bigl(C_B,(C_L)_L\bigr)
 =q(C_B)\prod_{L\in\mathcal L}
   \frac{p_L(C_L,C_B)}{q(C_B)}.
 \label{eq:gluing}
\end{equation}
Set $\mu=0$ when $q(C_B)=0$.  The independent choices here are the outcomes
for different $A$ settings conditional on the common $B$ outcome.  Each
factor retains the supplied dependence of its $A$ outcome on $C_B$.

For fixed $C_B$ with positive weight, each conditional factor sums to one.
Summing $\mu$ over all $A$-cosets therefore leaves $q(C_B)$, and summing over
$C_B$ gives total weight one.  To recover the distribution in a particular
context $K_L$, fix $(C_L,C_B)$ with $q(C_B)>0$ and sum over the $A$-cosets
for all other contexts:
\begin{equation*}
 \begin{aligned}
  &\sum_{(C_{L'})_{L'\ne L}}
    \mu\bigl(C_B,(C_{L'})_{L'\in\mathcal L}\bigr)\\
  &\qquad=q(C_B)\frac{p_L(C_L,C_B)}{q(C_B)}
    =p_L(C_L,C_B).
 \end{aligned}
\end{equation*}
If $q(C_B)=0$, every $p_L(C_L,C_B)$ vanishes, so the zero-weight
prescription also recovers these probabilities.  Thus $\mu$ reproduces each
supplied joint distribution, including its correlations between $A$ and $B$,
and is a noncontextual model of the selected family.

By Eq.~\eqref{eq:deterministic-values}, every deterministic assignment of the
complete Weyl scenario is labelled by a phase-space point $x=(x_A,x_B)$.
Its outcomes in the selected contexts are $C_L=x_A+L$ and $C_B=x_B+S$, so
Eq.~\eqref{eq:common-point} is necessary.  Conversely, if that condition
holds, choose any representative $x_B$ of $C_B$.  The complete-scenario
assignment labelled by $(x_A,x_B)$ has all the required selected outcomes.
This proves sufficiency and leaves $C_B$ unrestricted.
\end{proof}

On the selected contexts alone, fix $C_B$ and choose $C_L\neq u+L$ for
every $L$.  This is an allowed deterministic assignment and makes every
target event $E_L$ false.  Since event probabilities are nonnegative, the
minimum of $\sum_{L\in\mathcal L}p(E_L)$ over selected-context deterministic
assignments is therefore zero.  Among assignments that extend to the complete
scenario, the minimum is one, as established in the proof of
Theorem~\ref{thm:marginal-lbi}.  The probability sum is linear in mixtures,
so these are also the respective minima over noncontextual models on the
selected and complete scenarios.

\subsection{Relation to graph witnesses}
\label{sec:graph-witnesses}

To decide whether several joint outcomes can belong to one deterministic
assignment, we compare the values they give to shared observables.  Whenever
two outcomes disagree, choosing one rules out the other.  We can represent
these exclusions by a graph and use its independence number to bound the
sum of event probabilities in a noncontextual model
\cite{CabelloSeveriniWinter2014}.  For our inequality, we follow the
correspondence in Ref.~\cite[Theorem~1]{deSilva2017} and take one vertex for
each context-labelled complete outcome in $\phi_L^{(u)}$ or $\phi_{K_j}$,
joining pairs that disagree on a shared Weyl observable.  Choosing vertices
with no edges between them---an independent set---then gives outcomes that
agree on overlaps and define a consistent assignment on the observables
they cover.

Summing the vertex probabilities gives the left-hand side of
Eq.~\eqref{eq:full-logical-bell}.  Since distinct outcomes of one context
are joined, an independent set contains at most one outcome per context.
Choosing one from each of the $N=D_A+1+R$ contexts would satisfy every
proposition.  Theorem~\ref{thm:marginal-lbi} rules this out, while the
phase-space assignment attaining its bound supplies $N-1$ vertices.
We therefore recover $N-1=D_A+R$ as the graph's independence number.
To reduce this graph inequality to the marginal expression, we again use
the probability-one constraints of the complete scenario.

For a single-qudit marginal, the complementary outcomes used here also
appear in the two-qudit construction of Ref.~\cite{HowardEtAl2014}.
We set $n=2$ and $k=1$ to compare the two witnesses.  In that construction,
the authors pair the complementary single-qudit projectors with
computational-basis projectors on the second qudit and add all entangled
two-qudit stabilizer projectors.  Their combined sum, expressed using the
spread identity $A_u=\sum_L\Pi_{u+L}-I_A$, is
\begin{equation}
 \Sigma_H=(d^3I_A-A_u)\otimes I_B,
 \label{eq:howard-operator}
\end{equation}
and the orthogonality graph of these projectors has independence number $d^3$
\cite{HowardEtAl2014}.  Since $\Sigma_H$ acts as the identity on $B$, its
expectation depends only on the marginal $\sigma$.  Thus, for every
two-qudit state $\rho$ with this marginal, we obtain
\begin{equation}
 \begin{aligned}
  \Tr(\Sigma_H\rho)-d^3
  &=-\Tr(A_u\sigma)\\
  &=1-\sum_{L\in\mathcal L}p_\rho(E_L).
 \end{aligned}
 \label{eq:howard-comparison}
\end{equation}
The two witnesses therefore give the same quantum excess above their
respective noncontextual bounds.  They impose consistency through different
event families: the earlier construction uses a two-qudit orthogonality
graph, whereas we include valid-outcome propositions for the remaining
Weyl contexts.  Our logical construction applies to a block of any size
$1\leq k<n$ and gives its contextual-fraction bound in the complete scenario.

\section{Product states attaining the marginal bound}
\label{sec:exact-cf}

Consider a family of states $\sigma_\delta$ on $A$ whose Wigner functions
are nonpositive at a fixed phase-space point $u\in V_A$ and take the same
positive value at every other point.  Writing
$W_{\sigma_\delta}(u)=-\delta/D_A$, with $0\leq\delta\leq1$, normalization gives
\begin{equation}
 W_{\sigma_\delta}(a)=
 \begin{cases}
  -\delta/D_A,&a=u,\\[2pt]
  (D_A+\delta)/[D_A(D_A^2-1)],&a\neq u.
 \end{cases}
 \label{eq:single-facet-wigner}
\end{equation}
The trace orthogonality in Eq.~\eqref{eq:phase-point-identities} gives the
inverse expansion $\sigma_\delta=\sum_a W_{\sigma_\delta}(a)A_a$, hence
\begin{equation}
 \sigma_\delta=
 \frac{(D_A+\delta)I_A-(1+D_A\delta)A_u}{D_A^2-1}.
 \label{eq:delta-state}
\end{equation}
The operator is Hermitian and has trace one.  Its eigenvalues on the $+1$
and $-1$ eigenspaces of the displaced parity operator $A_u$
\cite{Gross2006} are $(1-\delta)/(D_A+1)$ and
$(1+\delta)/(D_A-1)$, respectively.  Both are nonnegative for
$0\leq\delta\leq1$, so the profile defines a physical state throughout this
interval.

At $\delta=0$, the Wigner value at $u$ vanishes and all other values are
positive.  The corresponding state is
\begin{equation}
 \sigma_0=\frac{D_A I_A-A_u}{D_A^2-1}.
 \label{eq:boundary-state}
\end{equation}
Denote the state at $\delta=1$ by
\begin{equation*}
 \alpha_-=\sigma_1=\frac{I_A-A_u}{D_A-1}.
\end{equation*}
The operator $(I_A-A_u)/2$ is the projector onto the $-1$ eigenspace of
$A_u$, of rank $(D_A-1)/2$.  Thus $\alpha_-$ is this projector divided by
its rank: it is proportional to the identity within that eigenspace and
vanishes on the orthogonal complement.

\begin{samepage}
\begin{theorem}
\label{thm:exact-cf}
Let $d$ be an odd prime, let $n=k+m$ with $k,m\geq1$, and let
$0\leq\delta\leq1$.  For every normalized state $\eta$ on $B$ with a
nonnegative Gross Wigner function, the complete $n$-qudit Weyl scenario
satisfies
\begin{equation}
 \CF(e_{\sigma_\delta\otimes\eta})
 =\delta=-\Tr(A_u\sigma_\delta)
 =-D_A W_{\sigma_\delta}(u).
 \label{eq:exact-product-cf}
\end{equation}
\end{theorem}
\end{samepage}

\begin{proof}
The quantum-state decomposition
\begin{equation}
 \sigma_\delta\otimes\eta
 =(1-\delta)(\sigma_0\otimes\eta)
  +\delta(\alpha_-\otimes\eta)
 \label{eq:product-decomposition}
\end{equation}
also decomposes the full probability model, since Born probabilities are
linear in the state.  Equation~\eqref{eq:single-facet-wigner} at
$\delta=0$ and Eq.~\eqref{eq:tensor-factorization} show that
$\sigma_0\otimes\eta$ has a nonnegative Gross Wigner function.  It therefore
has a noncontextual model for every maximal context in the complete Weyl
scenario, including contexts that do not factor across the $A|B$ split;
nonmaximal outcomes follow by coarse-graining.  The first term in
Eq.~\eqref{eq:product-decomposition} supplies a noncontextual component of
weight $1-\delta$, so
\begin{equation*}
 \CF(e_{\sigma_\delta\otimes\eta})\leq\delta.
\end{equation*}
The $A$-marginal is $\sigma_\delta$.  Corollary~\ref{cor:marginal-cf-bound}
and Eq.~\eqref{eq:single-facet-wigner} give the lower bound
\begin{equation*}
 \CF(e_{\sigma_\delta\otimes\eta})
 \geq-D_A W_{\sigma_\delta}(u)=\delta.
\end{equation*}
The two bounds agree.
\end{proof}

At $\delta=0$, the product model is noncontextual; at $\delta=1$,
$\alpha_-\otimes\eta$ has contextual fraction one.  For every $\delta$ in
this interval, the product state also minimizes the contextual fraction among all
joint states with the same marginals.  Indeed,
Corollary~\ref{cor:marginal-cf-bound} gives $\CF(e_\rho)\geq\delta$ for every
state with $A$-marginal $\sigma_\delta$, and
Theorem~\ref{thm:exact-cf} supplies an attaining state with $B$-marginal
$\eta$.  Hence, for $0\leq\delta\leq1$ and $\eta$ as in the theorem,
\begin{equation*}
 \min_{\substack{\rho:\ \Tr_B\rho=\sigma_\delta\\
                        \Tr_A\rho=\eta}}
 \CF(e_\rho)=\delta,
\end{equation*}
where the minimum is over quantum states, including correlated extensions.

Wigner sum negativity measures the total negative weight of the Wigner
function \cite{VeitchEtAl2014}:
\[
 \sn(\rho)=\sum_{a:W_\rho(a)<0}\lvert W_\rho(a)\rvert.
\]
Since $\sigma_\delta$ has only one negative Wigner value for $\delta>0$,
its sum negativity and the exact contextual fraction of the product states
in Theorem~\ref{thm:exact-cf} obey
\begin{align}
 \sn(\sigma_\delta)&=\frac{\delta}{D_A},
 \label{eq:single-facet-magic}\\
 \CF(e_{\sigma_\delta\otimes\eta})&=D_A\sn(\sigma_\delta).
 \label{eq:cf-sum-negativity}
\end{align}
Both quantities vanish at $\delta=0$.

When $A$ is one qutrit, the negative-parity eigenspace is one-dimensional,
so $\alpha_-$ is pure.  At $u=0$, it is
$\alpha_-=|\mathrm S\rangle\langle\mathrm S|$, where the strange state
\begin{equation}
 |\mathrm S\rangle=\frac{|1\rangle-|2\rangle}{\sqrt2}
 \label{eq:strange-state}
\end{equation}
has Wigner value $-1/3$ at the origin and $1/6$ at every other point
\cite{VeitchEtAl2014,JainPrakash2020}.  Other choices of $u$ give its
displaced versions.  Mixing this state with the maximally mixed state
gives the depolarized strange-state family
\begin{equation}
 \sigma(p)=(1-p)\frac{I_A}{3}+p\alpha_-,
 \qquad 0\leq p\leq1,
 \label{eq:qutrit-strange-ray}
\end{equation}
where $p$ is the weight of the strange state and $1-p$ is the depolarizing-noise
fraction on $A$.  For $0\leq p\leq1/4$, the state lies between $I_A/3$ and
$\sigma_0$, both of which are Wigner-nonnegative.  Its product with any
$\eta$ as in Theorem~\ref{thm:exact-cf} therefore has contextual
fraction zero.  For $1/4\leq p\leq1$, it equals $\sigma_\delta$ with
$\delta=(4p-1)/3$, and the theorem gives contextual fraction $\delta$.
Combining the two intervals gives the exact dependence on depolarizing noise:
\begin{equation*}
 \CF(e_{\sigma(p)\otimes\eta})
 =\left[\frac{4p-1}{3}\right]_+,
 \qquad 0\leq p\leq1.
\end{equation*}
The contextual fraction becomes nonzero above $p=1/4$ and increases linearly
to one at $p=1$.  Each of the four target events $E_L$ in
Eq.~\eqref{eq:qutrit-test} has probability $(1-p)/3$, so their sum is $4(1-p)/3$.
Its deficit below the noncontextual bound of one is exactly the contextual
fraction above threshold.  The product-state decomposition proves that
this measured lower bound accounts for all the contextuality in the
complete scenario.

\section{Discussion}
\label{sec:discussion}

In this article, we have constructed logical Bell inequalities that turn
marginal Wigner negativity into a lower bound on the contextual fraction of
a composite state in the complete Weyl scenario.  Since the probabilities
in our inequalities depend only on a reduced state, the bound applies to
every joint preparation with that marginal, irrespective of its correlations
with the complementary subsystem.

The available reduced states determine which marginal witnesses can be
evaluated.  A larger subsystem can retain negativity that disappears from
its smaller marginals, while the dimension scaling can make a smaller
subsystem's witness stronger.  The examples in
Section~\ref{sec:subsystem-choice} exhibit both possibilities.  Thus the
useful measurement target is set by the negative Wigner values and the
associated probability sums, rather than by subsystem size alone.

Certifying contextuality from these probabilities relies on the relations
between measurements in the complete
Weyl scenario.  The selected contexts admit a noncontextual
model for every compatible probability table.  A noncontextual explanation
of the complete scenario must also assign the same value to each shared
observable and respect the valid joint-outcome relations, including products
of commuting observables.  For Weyl measurements, these valid-outcome
relations hold with certainty for every quantum state.  Their probability-one
contributions allow us to evaluate the bound using only the selected-context
probabilities.

An attained lower bound identifies the least contextual fraction compatible
with the known reduced states.  The product family in
Section~\ref{sec:exact-cf} gives such a minimum for any normalized
Wigner-nonnegative state of the other subsystem.  Within this family,
changing that state changes joint measurement statistics while leaving the
contextual fraction fixed by the chosen marginal.

How much can correlations change this value while preserving both reduced
states?  Related continuous-variable examples have identical
Wigner-nonnegative marginals but different global Wigner negativities,
even for separable states \cite{TaghiabadiEtAl2016}.
The preparations in Appendix~\ref{app:fixed-marginals} have the
same distributions of the local displaced-parity signs.  Independent
sampling attains the minimum allowed by the marginals, while pairing
opposite signs attains contextual fraction one.  For $0\leq\delta<1$,
their contextual fractions differ even though both preparations are
separable.  The undetermined contextuality can therefore reside in these
correlations without requiring entanglement across the partition.

For more general reduced states, an open question is when the marginal
inequalities determine the least contextual fraction compatible with them.
Several negative Wigner values, possibly from different subsystems, give
constraints on the same probability model.  A common contextual component
can account for several violations, so their sum is not generally a lower
bound on its weight.  Optimizing these constraints together may improve on
the largest individual violation.  Determining when a physical joint state
attains that stronger bound would extend the exact marginal inference
beyond the product family considered here.

Joint measurements offer a way to distinguish preparations whose reduced
states coincide.  In our examples, the distinction is expressed by the
correlation between local parity signs.  We would like to know which additional joint
measurements, together with the known marginals, suffice to determine the
contextual fraction.

The spread supplies a phase-space cover, and the projector
identity makes the inequality test the sign of a chosen marginal Wigner
value.  Analogous geometrical constructions in other measurement scenarios
could target facets of the noncontextual probability polytope, the convex
hull of consistent deterministic outcome assignments.  The completeness
theorem in Ref.~\cite[Theorem~VI.5]{AbramskyHardy2012} characterizes these
polytopes using general logical Bell inequalities with nonnegative integer
weights.  The question is which facets admit simple
constructions from the geometry of measurement outcomes, and which quantum
preparations violate the resulting inequalities.

The connection between contextuality and magic-state computation gives a
physical motivation for understanding how contextuality behaves in composite
systems \cite{HowardEtAl2014}.  Resource states may form part of larger
preparations, where their properties must be understood together with the
correlations they share with other subsystems.  Our results give quantitative
meaning to this relationship through the information contained in reduced
states.  This offers a way to investigate how the contextuality relevant to
quantum computation depends on subsystem resources and their correlations.

\begin{acknowledgments}
S.S.\ would like to thank Dr.~Amolak Ratan Kalra and Dr.~Yash Totani
for fruitful discussions which guided the results, and Avantika Agarwal
for useful technical input.

This work was supported by the Natural Sciences and Engineering Research
Council of Canada. Wilfrid Laurier University is located in the traditional
territory of the Neutral, Anishnawbe and Haudenosaunee peoples. We thank
them for allowing us to conduct this research on their land.
\end{acknowledgments}

\appendix

\section{Fixed marginals and displaced-parity correlations}
\label{app:fixed-marginals}

\subsection{Preparations with the same reduced states}

Consider mixtures of fixed local states with definite displaced parity.
The probabilities of the local signs determine the reduced states, while
their joint distribution determines the correlations between the
subsystems.  We use the state $\sigma_\delta$ of
Section~\ref{sec:exact-cf} on $A$ and choose a state on $B$ whose parity
probabilities allow the signs to be sampled either independently or in
opposite pairs.

Let $d$ be an odd prime and $n=k+m$ with $k,m\geq1$, and write
$D_A=d^k$ and $D_B=d^m$ for the subsystem dimensions.  Retain the point $u\in V_A$ specifying
$\sigma_\delta$ and choose $v\in V_B$; both points remain fixed throughout
this appendix.  The displaced-parity operators $A_u$ and $A_v$ have
eigenvalues $+1$ and $-1$.  The state $\alpha_-$ from
Section~\ref{sec:exact-cf} is maximally mixed within the $-1$ eigenspace of
$A_u$.  Define the corresponding state for the $+1$ eigenspace by
\begin{equation}
 \alpha_+=\frac{I_A+A_u}{D_A+1}.
 \label{eq:parity-sector-states}
\end{equation}
On $B$, use the two states
\begin{equation}
 \beta_+=\frac{I_B+A_v}{D_B+1},
 \qquad
 \beta_-=\frac{I_B-A_v}{D_B-1}.
 \label{eq:beta-sector-states}
\end{equation}
The projectors $(I_A\pm A_u)/2$ and $(I_B\pm A_v)/2$ have ranks
$(D_A\pm1)/2$ and $(D_B\pm1)/2$, respectively, by
Eq.~\eqref{eq:phase-point-identities}.  Dividing each projector by its rank
gives these normalized states.  Preparing $\alpha_+$ or $\alpha_-$ therefore
gives the corresponding $A_u$ parity sign with certainty, and the same
holds for $\beta_+$ and $\beta_-$ on $B$.

The probability of a parity sign is the trace of $\sigma_\delta$ over the
corresponding eigenspace.  For $0\leq\delta\leq1$, multiplying the eigenvalues
from Section~\ref{sec:exact-cf} by the eigenspace ranks gives
\begin{equation*}
 p_+(\delta)=\frac{1-\delta}{2},
 \qquad
 p_-(\delta)=\frac{1+\delta}{2}.
\end{equation*}
The subscript on $p_\pm$ refers to the sign on $A$.  Choose the $B$ marginal
so that its positive sign has probability $p_-(\delta)$ and its negative
sign has probability $p_+(\delta)$.  The prescribed reduced states are then
\begin{equation}
 \begin{aligned}
  \sigma_\delta&=p_+(\delta)\alpha_++p_-(\delta)\alpha_-,\\
  \tau_\delta&=p_-(\delta)\beta_++p_+(\delta)\beta_-.
 \end{aligned}
 \label{eq:sector-mixtures}
\end{equation}
At $\delta=0$, the two signs on $B$ are equally likely and
\begin{equation}
 \tau_0=\frac{\beta_++\beta_-}{2}
 =\frac{D_B I_B-A_v}{D_B^2-1}.
 \label{eq:tau-boundary-state}
\end{equation}
The same family can therefore be written as
\begin{equation}
 \tau_\delta=(1-\delta)\tau_0+\delta\beta_+.
 \label{eq:tau-delta}
\end{equation}
Evaluating its Wigner function with
Eqs.~\eqref{eq:gross-definitions} and \eqref{eq:phase-point-identities}
gives
\begin{equation}
 W_{\tau_\delta}(b)=
 \begin{cases}
  \delta/D_B,&b=v,\\[2pt]
  (D_B-\delta)/[D_B(D_B^2-1)],&b\neq v.
 \end{cases}
 \label{eq:tau-delta-wigner}
\end{equation}
Thus $\tau_\delta$ has a nonnegative Gross Wigner function throughout the
interval.  In particular, $\tau_0$ has a zero Wigner value at $v$ and equal
positive values elsewhere.

Sampling the two local signs independently and preparing the corresponding
states $\alpha_\pm$ and $\beta_\pm$ gives the product of these reduced
states.  Alternatively, sample the sign on $A$ with probabilities
$p_\pm(\delta)$ and always choose the opposite sign on $B$.  These two
choices give the joint states
\begin{equation}
 \begin{aligned}
  \rho_{\mathrm{prod}}(\delta)&=\sigma_\delta\otimes\tau_\delta,\\
  \rho_{\mathrm{corr}}(\delta)&=
    p_+(\delta)\alpha_+\otimes\beta_-\\
   &\quad+p_-(\delta)\alpha_-\otimes\beta_+.
 \end{aligned}
 \label{eq:equal-marginal-states}
\end{equation}
Each local state in these mixtures has trace one, so tracing out either
subsystem recovers the corresponding marginal in
Eq.~\eqref{eq:sector-mixtures}.  Both preparations are convex mixtures of
product states and hence separable.  Since each local state is proportional
to a parity projector, the two joint states are also diagonal in the
same product eigenbasis of $A_u$ and $A_v$.

The mean local parities are the differences between the positive- and
negative-sign probabilities.  For the product preparation, independence
makes the joint parity expectation the product of these means.  In the
correlated preparation, every term has opposite signs and hence joint
parity $-1$.  The expectations are
\begin{subequations}
\label{eq:parity-expectations}
\begin{equation}
 \Tr(A_u\sigma_\delta)=-\delta,
 \qquad
 \Tr(A_v\tau_\delta)=\delta,
\end{equation}
\begin{equation}
 \begin{aligned}
  \Tr[(A_u\otimes A_v)\rho_{\mathrm{prod}}(\delta)]&=-\delta^2,\\
  \Tr[(A_u\otimes A_v)\rho_{\mathrm{corr}}(\delta)]&=-1.
 \end{aligned}
 \label{eq:perfect-opposition}
\end{equation}
\end{subequations}

\subsection{Contextual fractions at fixed marginals}

Equation~\eqref{eq:tensor-factorization} identifies the joint parity
$A_u\otimes A_v$ with the full-system phase-point operator $A_{u\oplus v}$.
Its expectation can therefore be obtained from a spread of joint Weyl
measurements on all $n$ qudits, using Eq.~\eqref{eq:spread-projector-sum}.

\begin{theorem}
\label{thm:fixed-marginals}
Let $d$ be an odd prime, let $n=k+m$ with $k,m\geq1$, and let
$0\leq\delta\leq1$.  Fix $u\in V_A$ and $v\in V_B$, and define
$\sigma_\delta$, $\tau_\delta$, $\rho_{\mathrm{prod}}(\delta)$ and
$\rho_{\mathrm{corr}}(\delta)$ as above.  Both joint states are separable,
diagonal in a common product eigenbasis of $A_u$ and $A_v$, and have
$A$-marginal $\sigma_\delta$ and $B$-marginal $\tau_\delta$.
Among all quantum joint states $\rho$ with these ordered marginals, the
complete $n$-qudit Weyl scenario satisfies
\begin{equation}
 \begin{aligned}
  \min_\rho\CF(e_\rho)
   &=\CF(e_{\rho_{\mathrm{prod}}(\delta)})=\delta,\\
  \max_\rho\CF(e_\rho)
   &=\CF(e_{\rho_{\mathrm{corr}}(\delta)})=1.
 \end{aligned}
 \label{eq:fixed-marginal-range}
\end{equation}
\end{theorem}

\begin{proof}
The marginal and separability statements follow from the construction in
Eqs.~\eqref{eq:sector-mixtures} and \eqref{eq:equal-marginal-states}.
For every quantum state $\rho$ with the prescribed marginals,
Corollary~\ref{cor:marginal-cf-bound} and
Eq.~\eqref{eq:single-facet-wigner} give
\begin{equation*}
 \CF(e_\rho)\geq-D_A W_{\sigma_\delta}(u)=\delta.
\end{equation*}
The state $\tau_\delta$ is normalized and Wigner-nonnegative by
construction.  Theorem~\ref{thm:exact-cf} with $\eta=\tau_\delta$ therefore
gives $\CF(e_{\rho_{\mathrm{prod}}(\delta)})=\delta$, proving the stated
minimum.

For the correlated preparation, choose a symplectic spread
$\mathcal L_n$ of $V_n$, independently of the spread on $A$ used in the
marginal inequality.  For each context $J\in\mathcal L_n$, let $E_J$ be
the coset outcome $(u\oplus v)+J$.  Let $\phi_J^{(u\oplus v)}$ state that the outcome
is any other valid coset:
\begin{equation*}
 \phi_J^{(u\oplus v)}=
 \bigvee_{\substack{C\in V_n/J\\C\neq(u\oplus v)+J}}\phi_{J,C},
\end{equation*}
where $\phi_{J,C}$ is the conjunction of Weyl-eigenvalue statements in
Eq.~\eqref{eq:probability-one-proposition}.  Thus
$p(\phi_J^{(u\oplus v)})=1-p(E_J)$.  For every remaining maximal context $K$, include
the valid-outcome proposition $\phi_K$ from the same equation, with
$p(\phi_K)=1$.

Suppose a global truth assignment satisfied all the $\phi_J^{(u\oplus v)}$ and
$\phi_K$.  Each $\phi_J^{(u\oplus v)}$ is a disjunction of valid outcomes, so the
assignment would specify a valid joint outcome in every maximal context.
The conjunctions use the same Weyl-value propositions on overlapping
contexts and enforce the products of commuting observables.  Hence
Eq.~\eqref{eq:deterministic-values} would give a single phase-space point
$x\in V_n$ labelling this assignment.  To satisfy $\phi_J^{(u\oplus v)}$, its outcome
$x+J$ would have to differ from $(u\oplus v)+J$, or equivalently,
$x-(u\oplus v)\notin J$.  This cannot hold for every $J\in\mathcal L_n$,
since the spread covers $V_n$.  The full family is therefore jointly
unsatisfiable.

Let $N$ be the number of propositions in this family, one for each maximal
context.  Equation~\eqref{eq:general-logical-bell} gives the full inequality
\begin{equation}
 \sum_{J\in\mathcal L_n}p(\phi_J^{(u\oplus v)})
 +\sum_{K\notin\mathcal L_n}p(\phi_K)\leq N-1,
 \label{eq:full-system-logical-bell}
\end{equation}
where the second sum runs over the remaining maximal Weyl contexts.
A deterministic assignment labelled by $x\neq u\oplus v$ makes exactly one
$\phi_J^{(u\oplus v)}$ false and satisfies all the other propositions, so it attains
this bound.  There are $d^n+1$ spread contexts.  Substituting
$p(\phi_J^{(u\oplus v)})=1-p(E_J)$ and $p(\phi_K)=1$ therefore reduces the inequality
to $\sum_{J\in\mathcal L_n}p(E_J)\geq1$.

For the correlated state, the spread-projector identity
\eqref{eq:spread-projector-sum} and the joint parity expectation
\eqref{eq:perfect-opposition} give
\begin{equation}
 \begin{aligned}
  &\sum_{J\in\mathcal L_n}p_{\rho_{\mathrm{corr}}(\delta)}(E_J)\\
  &\quad=\Tr[(I+A_{u\oplus v})\rho_{\mathrm{corr}}(\delta)]\\
  &\quad=0.
 \end{aligned}
 \label{eq:full-system-zero-probability-sum}
\end{equation}
Each target-outcome probability is nonnegative, so each must vanish.
Every $\phi_J^{(u\oplus v)}$ therefore has probability one, as does every $\phi_K$.
The full probability sum in Eq.~\eqref{eq:full-system-logical-bell} equals
$N$, its algebraic upper bound.  The normalized logical Bell violation is
thus one.

As in the proof of Corollary~\ref{cor:marginal-cf-bound}, consider any
admissible decomposition of the probability model,
\begin{equation*}
 e_{\rho_{\mathrm{corr}}(\delta)}
  =(1-\lambda)e_{\mathrm{NC}}+\lambda e',
 \qquad 0\leq\lambda\leq1.
\end{equation*}
The full sum is at most $N-1$ for the noncontextual component and at most
$N$ for $e'$.  Its value on the correlated preparation therefore obeys
\begin{equation*}
 N\leq(1-\lambda)(N-1)+\lambda N=N-1+\lambda.
\end{equation*}
Thus every admissible decomposition has $\lambda=1$.  This proves
$\CF(e_{\rho_{\mathrm{corr}}(\delta)})=1$.  Since contextual fraction is
at most one for every probability model, the correlated preparation attains
the maximum over joint states with the prescribed marginals.
\end{proof}

For $0\leq\delta<1$, the two preparations have different contextual
fractions despite having the same ordered pair of reduced states.  At
$\delta=0$, both reduced states have nonnegative Wigner functions, and the
product state is globally Wigner-nonnegative with contextual fraction zero.
The correlated state instead has
\begin{equation*}
 W_{\rho_{\mathrm{corr}}(0)}(u,v)=-\frac{1}{D_A D_B}
\end{equation*}
and contextual fraction one.  At $\delta=1$, the parity signs are certain:
negative on $A$ and positive on $B$.  Independent sampling and opposite
pairing then give the same state $\alpha_-\otimes\beta_+$.

The preparation with contextual fraction one uses the particular
complementary marginal $\tau_\delta$.  For any normalized Wigner-nonnegative
quantum state $\eta$ on $B$,
Section~\ref{sec:exact-cf} establishes that the minimum contextual fraction
compatible with $(\sigma_\delta,\eta)$ is $\delta$, attained by their
product.  If $\eta$ is pure, every compatible joint state must be that
product state, so its contextual fraction is fixed at $\delta$.

\bibliographystyle{apsrev4-2}
\bibliography{main}

\end{document}

%% file: figures/measurement-setup.tex
\begingroup
\begin{tikzpicture}[x=0.60cm,y=0.60cm,font=\footnotesize,
  every node/.style={text=black}]
  \begin{scope}[xshift=1.4cm]
    \node[anchor=west] at (-2.1,2.2) {(a)};
    \draw[black,line width=0.7pt] (-1.25,0)--(1.25,0);
    \foreach \q in {-1,1}{
      \draw[black,dashed,line width=0.6pt] (-1.25,\q)--(1.25,\q);
    }
    \foreach \p in {-1,0,1}{\foreach \q in {-1,0,1}{
      \fill[white] (\p,\q) circle[radius=2.5pt];
      \fill[black] (\p,\q) circle[radius=1.65pt];
    }}
    \draw[black,line width=0.65pt] (1,1) circle[radius=3.8pt];
    \node[anchor=south] at (1,1.3) {$x_A$};
    \foreach \q/\lab in {1/1,0/0,-1/2}{
      \node[anchor=east] at (-1.4,\q) {$q=\lab$};
    }
    \node[anchor=west] at (1.4,1) {$(0,1)+L_1$};
    \node[anchor=west] at (1.4,0) {$L_1$};
    \node[anchor=west] at (1.4,-1) {$(0,2)+L_1$};
    \foreach \p/\lab in {-1/2,0/0,1/1}{\node at (\p,-1.55) {$\lab$};}
    \node at (0,-2.15) {$p$};
  \end{scope}
  \begin{scope}[xshift=6.45cm]
    \node[anchor=west] at (-2.1,2.2) {(b)};
    \draw[black,line width=0.6pt] (-1.25,0)--(1.25,0);
    \draw[black,line width=0.6pt] (0,-1.25)--(0,1.25);
    \draw[black,line width=0.6pt] (-1.25,-1.25)--(1.25,1.25);
    \draw[black,line width=0.6pt] (-1.25,1.25)--(1.25,-1.25);
    \foreach \p in {-1,0,1}{\foreach \q in {-1,0,1}{
      \fill[white] (\p,\q) circle[radius=2.5pt];
      \fill[black] (\p,\q) circle[radius=1.65pt];
    }}
    \draw[black,line width=0.65pt] (1,1) circle[radius=3.8pt];
    \node[anchor=south west] at (1.25,1.25) {$L_3$};
    \node[anchor=west] at (1.35,0) {$L_1$};
    \node[anchor=north west] at (1.25,-1.25) {$L_4$};
    \node[anchor=south] at (0,1.35) {$L_2$};
    \node[anchor=north west,fill=white,inner sep=0.7pt] at (.1,-.1) {$u$};
  \end{scope}
\end{tikzpicture}
\endgroup